\documentclass[a4paper]{article}
\usepackage{CJK}
\usepackage[paperwidth=185mm,paperheight=230mm,textheight=190mm,textwidth=145mm,left=20mm,right=20mm, top=25mm, bottom=20mm]{geometry}
\usepackage[CJKbookmarks, colorlinks, bookmarksnumbered=true,pdfstartview=FitH,linkcolor=blue,citecolor=green]{hyperref}
\usepackage{amsmath,amssymb}
\usepackage{amsthm}
\usepackage{calc}
\usepackage{graphicx}
\usepackage{supertabular}
\usepackage{longtable}
\usepackage{float}
\usepackage{color}
\usepackage{enumerate}
\usepackage{colortbl,booktabs}
\usepackage{natbib}
\usepackage{authblk}
\usepackage{multirow}
\usepackage{float}

\newtheorem{theorem}{Theorem}
\newtheorem{lemma}{Lemma}

\newtheorem{remark}{Remark}

\def\wh{\widehat}
\def\wt{\widetilde}
\def\cov{\hbox{cov}}

\begin{document}

\title{ Functional principal component analysis estimator for non-Gaussian data }
\author[a]{{\fontsize{12pt}{18pt}\selectfont Rou Zhong}}
\author[a]{{\fontsize{12pt}{0.5em}\selectfont Shishi Liu}}
\author[b]{{\fontsize{12pt}{0.5em}\selectfont Haocheng Li}}
\author[a]{{\fontsize{12pt}{0.5em}\selectfont Jingxiao Zhang} \thanks{zhjxiaoruc@163.com}}
\affil[a]{{\emph\fontsize{12pt}{0.5em}\selectfont Center for Applied Statistics, School of Statistics, Renmin University of China}}
\affil[b]{{\emph\fontsize{12pt}{0.5em}\selectfont Department of Mathematics and Statistics, University of Calgary}}
\date{}
\maketitle

\begin{abstract}

Functional principal component analysis (FPCA) could become invalid when data involve non-Gaussian features. Therefore, we aim to develop a general FPCA method to adapt to such non-Gaussian cases. A Kenall's $\tau$ function, which possesses identical eigenfunctions as covariance function, is constructed. The particular formulation of Kendall's $\tau$ function makes it less insensitive to data distribution. We further apply it to the estimation of FPCA and study the corresponding asymptotic consistency. Moreover, the effectiveness of the proposed method is demonstrated through a comprehensive simulation study and an application to the physical activity data collected by a wearable accelerometer monitor.

\end{abstract}

\section{ Introduction }

With the rapid development of measurement tools and the advance of data collection technology, functional data become ubiquitous in various fields. Many representative monographs, such as \citet{ref1} and \citet{hsing2015theoretical} among others, provide comprehensive introduction on functional data analysis (FDA).

Functional principal component analysis (FPCA), as a crucial technique in FDA, is widely used in reducing dimensionality and exploring the major variation modes of sample curves.
There is an extensive literature on the work of FPCA. To list a few,
\citet{ref12} studied penalized spline models for FPCA. \citet{ref15} proposed a localized FPCA method to gain eigenfunctions with localized support regions. A dynamic FPCA approach was developed by \citet{ref13} with regard to functional time series. Moreover, \citet{chiou2014multivariate} and \citet{ref20} considered FPCA for multivariate functional data. Though a variety of FPCA approaches have been proposed, most of the existing FPCA methods involve inferences on the covariance structure. However, as pointed out in \citet{kraus2012dispersion}, the covariance structure inferences would lead to biased outcomes for non-Gaussian observations.

Meanwhile, non-Gaussian continuous data frequently occur in practice. For example, the physical activity data \citep{keadle}, which motivated our study, are shown to be non-Gaussian. The data were obtained from $63$ moderately overweight but healthy office workers in a health research project to increase exercise activities and reduce sedentary behaviors.
As shown in Section \ref{app}, the measurements bear some skewed characteristics,
which imply a deviation from Gaussian distribution.
The goal of this article is to develop adaptive FPCA estimators for such non-Gaussian cases.

Recently, the study of FPCA for non-Gaussian data has attracted increasing attention.
\citet{ref4} introduced a latent Gaussian process model for non-Gaussian longitudinal data that measured sparsely and irregularly. \citet{ref5} discussed a Bayesian FPCA approach for data in single-parameter exponential family. \citet{ref6} considered a generalized additive mixed model (GAMM) framework, and \citet{ref7} proposed an exponential family FPCA method. Nevertheless, all of these studies were developed based on the assumption that data are generated from a specified distribution. These methods may particularly work for discrete data (e.g. binary and count data), whereas not appropriate for the data we have at hand, as it is difficult to assume a proper distribution for the physical activity records.

In this paper, we propose a novel FPCA method, called Kendall FPCA, to handle such non-Gaussian measurements without any specific distributional assumption. We first introduce a Kendall's $\tau$ function, which can be seen as a functional generalization of the multivariate Kendall's $\tau$ matrix in \citep{han2018eca}. It can be further demonstrated that the Kendall's $\tau$ function shares identical eigenfunction with covariance function. Moreover, as rank information is involved in Kendall's $\tau$ function, the estimation is unlikely to be affected by data distribution. Consequently, more general estimate for eigenfunction can be achieved through the Kendall's $\tau$ function estimate. The computation is straightforward and quite easy to implement. Additionally, we further explore the asymptotic properties of our method.

The remainder of the article is organized as follows. In Section \ref{method}, we introduce the proposed FPCA method and further explore its theoretical properties. Section \ref{sim} reports the results from simulation studies. Section \ref{app} illustrates the application of our method to the physical activity data. Concluding discussions are provided in Section \ref{discussion}.

\section{Methodology}\label{method}

\subsection{Kendall FPCA}

Let $X(\cdot)$ denote a smooth random function in $L^2(\mathcal{T})$, with unknown mean function $E\{X(t)\} = \mu(t)$ and covariance function $\Gamma(s, t) = \cov\{X(s), X(t)\}$, where $\mathcal{T}$ is a bounded and closed interval. Define $\Gamma$ as a covariance operator that $(\Gamma \phi)(t) = \int \Gamma(s, t)\phi(s)ds$, with $\phi$ being an arbitrary function.
A commonly used approach to implement FPCA is through eigenanalysis of $\Gamma(s, t)$ \citep{ref1}.
In specific, the eigenequations,
\begin{align}
(\Gamma \phi_k)(t) = \lambda_k \phi_k(t), k=1, 2, \ldots, \nonumber
\end{align}
are considered, where $\phi_k$ and $\lambda_k$ are the $k$-th eigenfunction and eigenvalue respectively.
Further, according to the Karhunen-Lo\`eve theorem, $X(t)$ can be expanded by the eigenfunctions as follows:
\begin{align}\label{kl}
X(t) = \mu(t) + \sum_{k=1}^{\infty}\xi_k\phi_k(t),
\end{align}
where $\xi_k$ is the $k$-th functional principal component score with mean zero and variance $\lambda_k$.
Covariance function performs a crucial role in FPCA.
However, as the estimate of covariance function would become implausible when facing non-Gaussian cases \citep{kraus2012dispersion}, an appropriate alternative is demanded.

To this end, we construct a Kendall's $\tau$ function as follows:
\begin{align}
K(s, t) = E\Big [ \frac{\{X(s) - \wt{X}(s)\}\{X(t) - \wt{X}(t)\}}{\|X - \wt{X}\|^2} \Big ], \nonumber
\end{align}
where $\wt{X}$ is an independent copy of $X$, and $\|X - \wt{X}\|^2 = \int_{\mathcal{T}} \{X(t) - \wt{X}(t)\}^2 dt$. It is a functional generalization of the multivariate Kendall's $\tau$ matrix in \citep{han2018eca}. For the appealing properties of Kendall's $\tau$ function, it is shown in Lemma \ref{LemmaKend} that the eigenfunction of $K(s, t)$ is exactly the same with that of covariance function $\Gamma(s, t)$.

\begin{lemma}\label{LemmaKend}
The $k$-th eigenfunction $\phi_k$ for covariance function $\Gamma(s, t)$ is identical to the $k$-th eigenfunction for Kendall's $\tau$ function $K(s, t)$.
\end{lemma}

The proof of Lemma \ref{LemmaKend} can be referred to \citep{zhong2021robust}.
Consequently, eigenfunctions $\phi_k$'s satisfy the eigenequations
\begin{align}
\int_{\mathcal{T}} K(s, t) \phi_k(s) ds = \lambda_k^{\star} \phi_k(t), k = 1, 2, \ldots, \nonumber
\end{align}
where $\lambda_k^{\star}$ is the $k$-th eigenvalue of $K(s, t)$.

The proposed Kendall FPCA is based on eigenanalysis of the Kendall's $\tau$ function. In terms of the formulation of $K(s, t)$, it involves the comparison of two independent random processes, therefore contains some rank information and helps to resist negative affect of the skewed data. Moreover, $X - \wt{X}$ is projected onto a sphere with radius one, which would depress the influence of heavy-tailed data. These characteristics of Kendall's $\tau$ function alleviate its dependence on strict distributional assumptions.

\subsection{Estimation}

Suppose we observe a set of independent random curves $\{X_i(t): i=1, \ldots, N\}$ from $N$ individuals. The estimated Kendall's $\tau$ function is given by
\begin{align}\label{whK}
\wh{K}(s,t) = \frac{2}{N(N-1)} \sum_{i<j} \frac{ \{ X_i(s) - X_j(s) \} \{ X_i(t) - X_j(t) \} }{ \| X_i - X_j \|^2},
\end{align}
where $\|X_i-X_j\|^2=\int_{\mathcal{T}}\{X_i(t)-X_j(t)\}^2dt$. The above Kendall's $\tau$ function estimate is an operator-valued U-statistics and its asymptotic properties are discussed in Section \ref{SecTheory}.

Then the estimated eigenfunctions $\wh{\phi}_k$'s are obtained from the following eigenequations
\begin{align}
\int_{\mathcal{T}} \wh{K}(s, t) \wh{\phi}_k(s) ds = \wh{\lambda}_k^{\star} \wh{\phi}_k(t), k = 1, 2, \ldots, \nonumber
\end{align}
where $\wh{\lambda}_k^{\star}$ is the $k$-th eigenvalue estimate of $K(s, t)$. Further, we have
\begin{align}
\wh{\xi}_{ik} = \int_{\mathcal{T}} \{ X_i(t) - \wh{\mu}(t) \} \wh{\phi}_k(t) dt, \nonumber
\end{align}
where $\wh{\xi}_{ik}$ is the $k$-th estimated functional principal component score for the $i$-th individual, and $\wh{\mu}(t) = \sum_{i = 1}^N X_i(t)/N$ is the estimate for mean function. We can estimated the eigenvalues through variance estimate of the estimated functional principal component scores.

In practical, the individual function $X_i(t)$ might not be fully observed. Consider $X_i(t)$'s are recorded at a dense and regular grid of time points $t_1, \ldots, t_d$, then the smoothing technique is employed. The procedure is similar to the discretizing method in \citep{ref1}. We first compute the discretized $\wh{K}(s, t)$ by applying (\ref{whK}) to the $N$ by $d$ matrix of observed data, where $\|X_i-X_j\|^2$ is obtained by numerical integration. Then $\wh{\phi}_k(t)$ is estimated by smoothing the eigenvector of the discretized $\wh{K}(s, t)$, with further normalization such that $\int_{\mathcal{T}} \wh{\phi}_k^2 (t) dt = 1$.

\subsection{Theoretical results}\label{SecTheory}

In this section, we discuss asymptotic properties for both the estimated Kendall's $\tau$ function in (\ref{whK}) and the estimated eigenfunction $\wh{\phi}_k(t)$. As our proof in the Appendix illustrates, $\wh{K}(s, t)$ is an operator-valued U-statistics with kernel size two. Hence the convergence rate of our $\wh{K}(s, t)$ can be achieved via the asymptotic results of operator-valued U-statistics \citep{sriperumbudur2017approximate}.

\begin{theorem}\label{theoryKend}
For the estimated Kendall's $\tau$ function, we have
\begin{align}\label{supKend}
\sup_{s, t \in \mathcal{T}} |\wh{K}(s, t) - K(s, t)| = O_p(N^{-1/2}).
\end{align}
\end{theorem}

\begin{remark}
\normalfont Note that we assume the individual function is fully observed here. It is not a strict assumption as we take into account a dense design in this paper. With the increase of observation size, the smoothed individual curves can be quite close to the true curves so that the approximation error is ignorable.
\end{remark}

The following Theorem \ref{theoryeigfun} presents theoretical property of the estimated eigenfunction, which is a direct result from Theorem \ref{theoryKend}.

\begin{theorem}\label{theoryeigfun}
For the estimated eigenfunction, we have
\begin{align}
\sup_{t \in \mathcal{T}} |\wh{\phi}_k(t) - \phi_k(t)| = O_p(N^{-1/2}). \nonumber
\end{align}
\end{theorem}

\section{Simulation}\label{sim}

In this section, we conduct a simulation study to demonstrate the performance of our proposed Kendall FPCA under dense design.
We generate the data through Karhunen-Lo\`eve expansion with additional independent measurement errors to mimic the practical cases. In specific, the observed data can be expressed as
\begin{align}
Y_{ij} = \mu(t_j) + \sum_{k=1}^{\infty} \xi_{ik} \phi_k(t_j) + \epsilon_{ij}, \nonumber
\end{align}
where $\epsilon_{ij} \sim \mathcal{N}(0, \sigma^2)$ for $i = 1, \ldots, N, j = 1, \ldots, d$. We set mean function as $\mu(t) = 0, t \in [0, 10]$.
The eigenvalues are $\lambda_1 = 4^2$, $\lambda_2 = 3^2$ and $\lambda_k = 0$ for $k \geq 3$, with the eigenfunctions being as follows:
\begin{itemize}
\item Case 1: $\phi_1(t) = \mbox{cos}(\pi t/10)/\sqrt{5}$ and $\phi_2(t) = \mbox{sin}(\pi t/10)/\sqrt{5}$.
\item Case 2: $\phi_1(t) = \mbox{sin}(\pi t/5)/\sqrt{5}$ and $\phi_2(t) = \mbox{cos}(\pi t/5)/\sqrt{5}$.
\end{itemize}
Moreover, the variance $\sigma^2$ of measurement errors is $0.25$.

Similar to \citep{yao2005functional}, functional principal component score $\xi_{ik}$ is simulated with various distribution settings according to the corresponding eigenvalue $\lambda_k$. We consider Gaussian, mix-Gaussian, EC2 \citep{han2018eca}, skew-t \citep{azzalini2014the} distributions here.
The specific settings for these four distributions are as follows:
\begin{itemize}
\item For Gaussian distribution, $\xi_{ik}$ is sampled from $\mathcal{N}(0, \lambda_k)$.
\item For mix-Gaussian distribution, $\xi_{ik}$ is sampled from $\mathcal{N}(\sqrt{\lambda_k/2}, \lambda_k/2)$ and $\mathcal{N}(-\sqrt{\lambda_k/2}, \lambda_k/2)$ with equal probability.
\item For EC2 distribution, $\xi_{ik} = \sqrt{\lambda_k} \eta_{ik} u_{ik} / \sqrt{2}$, where $\eta_{ik}$ follows the exponential distribution with parameter 1, and $u_{ik}$ is chosen from 1 and -1 with equal probability.
\item For skew-t distribution, besides the mean and variance for $\xi_{ik}$ being set as $0$ and $\lambda_k$, the coefficients of skewness and excess kurtosis are given by $1.5$ and $5.1$ respectively.
\end{itemize}
Furthermore, EC2 distribution is heavy-tailed, while skew-t distribution is skewed.
In the simulation study of $100$ runs, we have $N = 100$ individuals and each subject has $d = 51$ equally-spaced observations.

As a comparison, the covariance-based FPCA approach \citep{ref1} and MLE-based FPCA approach \citep{peng2009a} are also explored. For simplicity, our method and these two methods are denoted as KFPCA, FPCA and MLE respectively.
The following criteria are evaluated to assess the performance of these three methods:
\begin{align}
\mbox{IMSE}_k &= \| \phi_k - \wh{\phi}_k \|^2 = \int_{\mathcal{T}} \{\wh{\phi}_k(t) - \phi_k(t)\}^2 dt, \nonumber \\
\mbox{MSE}_k &= \frac{1}{N} \sum_{i = 1}^N (\wh{\xi}_{ik} - \xi_{ik})^2, \nonumber
\end{align}
where $\mbox{IMSE}_k$ and $\mbox{MSE}_k$ measure the estimation error for the $k$-th eigenfunction and functional principal component score. Specifically, smaller $\mbox{IMSE}_k$ and $\mbox{MSE}_k$ suggest better performance.

Table \ref{simTableCase1} presents the simulation results under the four considered distributions in Case 1.
It is observed that FPCA, MLE and KFPCA methods yield similar results under Gaussian and mix-Gaussian distributions. Nevertheless, when the functional principal component scores turn to follow EC2 and skew-t distributions, it is evident that KFPCA performs much better than FPCA and MLE methods in terms of both $\mbox{IMSE}_k$ and $\mbox{MSE}_k$.
Moreover, $\mbox{IMSE}_k$ and $\mbox{MSE}_k$ of KFPCA under EC2 and skew-t distributions merely display a slight increase when compared with the Gaussian case, while FPCA and MLE methods exhibit an obvious deteriorating trend. The above numerical results indicate that the proposed KFPCA method produces more precise estimates for eigenfunctions and functional principal component scores in the non-Gaussian cases, especially when heavy-tailed and skewed features occur.

The simulation results for Case 2 are reported in Table \ref{simTableCase2}. For Gaussian and mix-Gaussian distributions, FPCA, MLE and KFPCA show slight difference with each other for both $\mbox{IMSE}_k$ and $\mbox{MSE}_k$, as what happens in Case 1. Although all these three methods become poorer for EC2 distribution, the proposed KFPCA provides best estimates for the eigenfunctions and functional principal component scores according to $\mbox{IMSE}_k$ and $\mbox{MSE}_k$. Considering skew-t distribution, KFPCA is much more accurate than FPCA and MLE regardless of $\mbox{IMSE}_k$ or $\mbox{MSE}_k$. Therefore, the effectiveness of our KFPCA method is further demonstrated.

\begin{table}[htbp]
\caption{Simulation results for the averaged $\mbox{IMSE}_1$, $\mbox{IMSE}_2$, $\mbox{MSE}_1$ and $\mbox{MSE}_2$ across $100$ runs in Case 1, with standard deviation in parentheses. The simulation settings and the definition of the $\mbox{IMSE}_1$, $\mbox{IMSE}_2$, $\mbox{MSE}_1$ and $\mbox{MSE}_2$ are illustrated in Section \ref{sim}.}
\label{simTableCase1}
\begin{center}
\begin{tabular}{ccccc}
\hline
 &Method&FPCA&MLE&KFPCA\\
\hline
\multirow{2}{*}{Gaussian}&$\mbox{IMSE}_1$&0.0352 (0.0530)&0.0352 (0.0529)&0.0347 (0.0581)\\
 &$\mbox{IMSE}_2$&0.0331 (0.0496)&0.0331 (0.0495)&0.0327 (0.0546)\\
 &$\mbox{MSE}_1$&0.5157 (0.5646)&0.5154 (0.5642)&0.5099 (0.6100)\\
 &$\mbox{MSE}_2$&0.6111 (0.6666)&0.6105 (0.6653)&0.6055 (0.7317)\\
\hline
\multirow{2}{*}{mix-Gaussian}&$\mbox{IMSE}_1$&0.0356 (0.0627)&0.0355 (0.0625)&0.0386 (0.0637)\\
 &$\mbox{IMSE}_2$&0.0334 (0.0588)&0.0333 (0.0586)&0.0362 (0.0598)\\
 &$\mbox{MSE}_1$&0.5638 (0.6561)&0.5635 (0.6550)&0.5926 (0.6674)\\
 &$\mbox{MSE}_2$&0.6301 (0.8543)&0.6292 (0.8517)&0.6705 (0.8644)\\
\hline
\multirow{2}{*}{EC2}&$\mbox{IMSE}_1$&0.0708 (0.2598)&0.0709 (0.2601)&0.0464 (0.1225)\\
 &$\mbox{IMSE}_2$&0.0686 (0.2580)&0.0687 (0.2583)&0.0440 (0.1184)\\
 &$\mbox{MSE}_1$&0.9455 (3.0307)&0.9464 (3.0333)&0.6476 (1.3408)\\
 &$\mbox{MSE}_2$&1.0149 (2.9945)&1.0162 (2.9988)&0.7304 (1.2933)\\
\hline
\multirow{2}{*}{Skew-t}&$\mbox{IMSE}_1$&0.0962 (0.2629)&0.0962 (0.2626)&0.0628 (0.1366)\\
 &$\mbox{IMSE}_2$&0.0925 (0.2574)&0.0925 (0.2572)&0.0594 (0.1307)\\
 &$\mbox{MSE}_1$&1.2380 (3.3262)&1.2375 (3.3187)&0.8548 (1.4715)\\
 &$\mbox{MSE}_2$&1.2829 (3.3437)&1.2821 (3.3346)&0.9181 (1.6226)\\
\hline
\end{tabular}
\end{center}
\end{table}

\begin{table}[htbp]
\caption{Simulation results for the averaged $\mbox{IMSE}_1$, $\mbox{IMSE}_2$, $\mbox{MSE}_1$ and $\mbox{MSE}_2$ across $100$ runs in Case 2, with standard deviation in parentheses. The simulation settings and the definition of the $\mbox{IMSE}_1$, $\mbox{IMSE}_2$, $\mbox{MSE}_1$ and $\mbox{MSE}_2$ are illustrated in Section \ref{sim}.}
\label{simTableCase2}
\begin{center}
\begin{tabular}{ccccc}
\hline
 &Method&FPCA&MLE&KFPCA\\
\hline
\multirow{2}{*}{Gaussian}&$\mbox{IMSE}_1$&0.0347 (0.0563)&0.0347 (0.0563)&0.0383 (0.0599)\\
 &$\mbox{IMSE}_2$&0.0374 (0.0597)&0.0374 (0.0598)&0.0413 (0.0636)\\
 &$\mbox{MSE}_1$&0.5379 (0.5739)&0.5377 (0.5740)&0.5696 (0.6025)\\
 &$\mbox{MSE}_2$&0.7086 (0.9653)&0.7083 (0.9657)&0.7744 (1.0440)\\
\hline
\multirow{2}{*}{mix-Gaussian}&$\mbox{IMSE}_1$&0.0303 (0.0565)&0.0303 (0.0565)&0.0319 (0.0566)\\
 &$\mbox{IMSE}_2$&0.0327 (0.0599)&0.0327 (0.0598)&0.0345 (0.0600)\\
 &$\mbox{MSE}_1$&0.4977 (0.5735)&0.4978 (0.5731)&0.5098 (0.5691)\\
 &$\mbox{MSE}_2$&0.6061 (0.8631)&0.6065 (0.8626)&0.6302 (0.8581)\\
\hline
\multirow{2}{*}{EC2}&$\mbox{IMSE}_1$&0.1248 (0.3005)&0.1249 (0.3008)&0.0941 (0.2003)\\
 &$\mbox{IMSE}_2$&0.1297 (0.3059)&0.1299 (0.3063)&0.0994 (0.2058)\\
 &$\mbox{MSE}_1$&1.5576 (3.5384)&1.5589 (3.5424)&1.1432 (2.1173)\\
 &$\mbox{MSE}_2$&1.8612 (3.8849)&1.8634 (3.8908)&1.5018 (2.5464)\\
\hline
\multirow{2}{*}{Skew-t}&$\mbox{IMSE}_1$&0.1498 (0.3495)&0.1496 (0.3490)&0.0560 (0.1151)\\
 &$\mbox{IMSE}_2$&0.1552 (0.3557)&0.1549 (0.3552)&0.0599 (0.1205)\\
 &$\mbox{MSE}_1$&1.8790 (4.2276)&1.8765 (4.2225)&0.7577 (1.3502)\\
 &$\mbox{MSE}_2$&2.2009 (4.6310)&2.1977 (4.6253)&0.9938 (1.5749)\\
\hline
\end{tabular}
\end{center}
\end{table}

\section{Real data analysis}\label{app}

In this section, we use the proposed method to analyze the physical activity dataset collected by wearable monitors \citep{keadle}. The raw activity signal is captured by a device named the ActivPAL$^{\rm TM}$ (www.paltech.plus.com). This device is taped in front of the thigh. It uses an accelerometer to measure the angle of the thigh and the movement speed. The raw data recorded by the accelerometer are used to calculate energy expenditure level in metabolic equivalents (METs). In general, physical activity with METs$<3$ can be identified as light activity; otherwise, it is moderate-to-vigorous activity. The dataset has $N=63$ subjects and each participant has $d = 36$ five-minute interval records.
Figure \ref{sample}(a) shows the energy expenditure observation from one subject across three hours. Figure \ref{sample}(b) shows the histogram of the observations from all subjects. Figures \ref{sample}(c)-(f) are histograms and Q-Q plots for the energy expenditure level at two time points. The plots suggest the observations are skewed and has extreme values.

\begin{figure}[H]
\centering
\includegraphics[width = \textwidth]{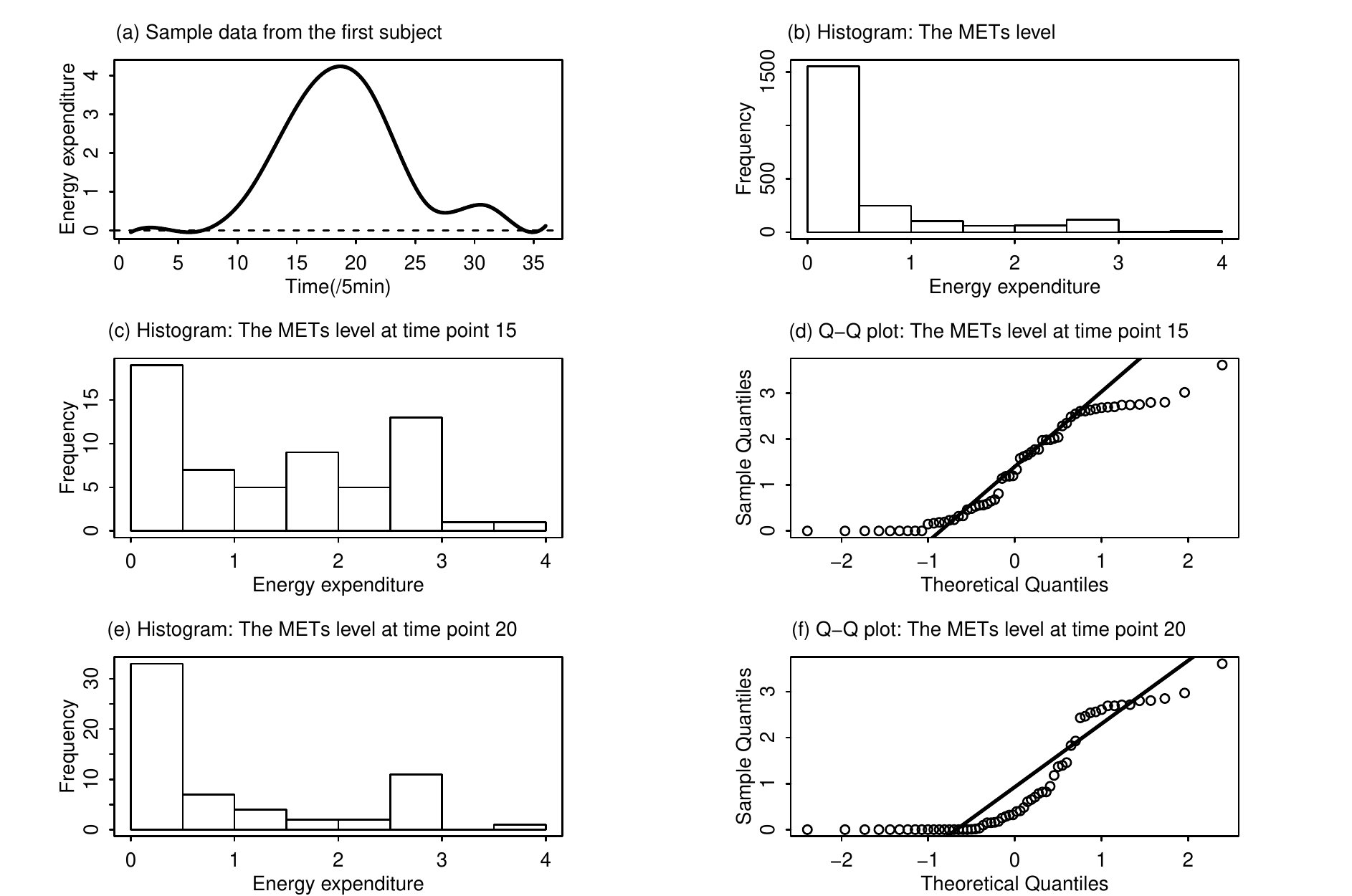}
\caption{ (a) Energy expenditure observation from subject with ID 1 across $3$ hours. (b) Histogram of the energy expenditure observations (METs) from all subjects. (c)(d) Histogram and Q-Q plot for the energy expenditure level (METs) at time point 15. (e)(f) Histogram and Q-Q plot for the energy expenditure level (METs) at time point 20.}
\label{sample}
\end{figure}

In this application, $X_i(t)$ can be defined as the energy expenditure level curve from subject $i$ $(i=1,\cdots,63)$.
Our proposed Kendall FPCA approach is used to analyze the physical activity dataset as the non-Gaussian features, which includes right skewness and extreme values, must be taken into account.
Figure \ref{realdata}(a) illustrates that mean energy expenditure level  slightly increases at $t=1$ and then decreases back to starting level at about $t=5$. The light energy expenditure could be explained as the warming up stage designed in this health research project to reduce the risk of injury. The METs level increases dramatically at about $t=12$, and decreases back to just above the starting level by $t=24$. During this interval, individuals take an intense exercise training for about 30 minutes and then start to cool down for about 25 minutes. The energy expenditure level stays the same after $t=24$, which shows subjects are taking sedentary activities after the exercise program.
Figures \ref{realdata}(b) displays the estimates of the first two eigenfunctions via our Kendall FPCA. The first eigenfunction rises to a peak at about $t=16$, and then decreases to zero. This suggests that main pattern of the variability in energy expenditure is around $t=16$ across different participants. Subjects would have highly different METs level during the intense exercise program because they have various physical training experience. On the other hand, the second eigenfunction has one positive and one negative peaks, which indicates energy expenditure variabilities at $t=15$ and $t=21$ are negatively correlated. That is, if a subject spends higher energy expenditure around $t=15$, he/she may have lower energy expenditure later.

\begin{figure}[H]
\centering
\includegraphics[width = \textwidth]{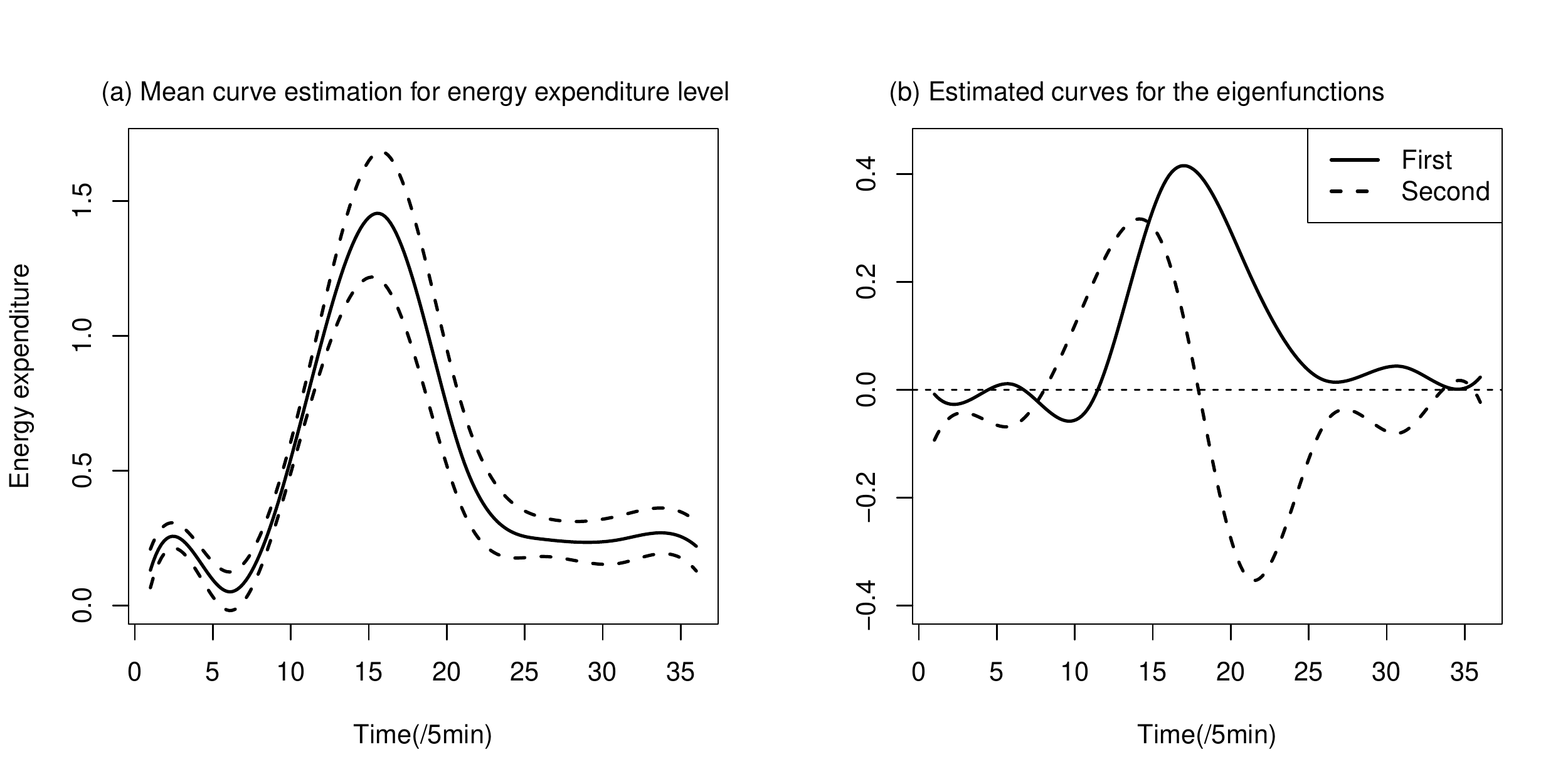}
\caption{FPCA results for the physical activity data illustrated in Section \ref{app}. (a) Mean curve estimation for energy expenditure level. Solid line and dashed line are the estimated curve and its 90\% bootstrap confidence intervals respectively. (b) Estimated curves for the first and the second eigenfunctions, in solid line and dashed line respectively.}
\label{realdata}
\end{figure}

\section{Discussion}\label{discussion}

In this paper, we develop the Kendall FPCA method to analyze non-Gaussian data. A Kendall's $\tau$ function is introduced, whose appealing properties indicate its potential in FPCA, especially for non-Gaussian data. The computation is easy-implemented and we also explore asymptotic properties of our method.
Further simulation study shows encouraging results of our method under both Gaussian and non-Gaussian settings.
Moreover, analysis of the physical activity data demonstrates utility of the developed method in real world data applications.

There are several extensions can be considered for the proposed method.
The functional data are assumed to have regular and dense structure in this paper. It is of interest to generalize the current method to irregular and sparse cases. For example, we may extend our methods to sparse cases by using the local linear estimation \citep{yao2005functional}. Moreover, we only consider continuous variable with non-Gaussian characteristics, it is also worthwhile to discuss FPCA approaches for discrete data, such as binary data and Poisson data.

\section*{Acknowledgements}

The research was supported by Public Health $\&$ Disease Control and Prevention, Major Innovation $\&$ Planning Interdisciplinary Platform for the ``Double-First Class" Initiative, Renmin University of China.
Li was supported by discovery grants program from the Natural Sciences and Engineering Research Council of Canada (NSERC, RGPIN-2015-04409). The authors thank Dr. Sarah Kozey Keadle for making the physical activity data available to them. Dr. Sarah Kozey Keadle was supported by a National Cancer Institute grant (R01-CA121005).

\appendix

\section{Proofs}

\begin{proof}[Proof of Theorem \ref{theoryKend}]
Define the operator corresponds to $K(s, t)$ as Kendall's $\tau$ operator and denote it as $K$. Then $K$ can be written as
\begin{align}
K = E \{S(X - \wt{X}) \otimes S(X - \wt{X})\}, \nonumber
\end{align}
where $S$ is the spatial sign map such that $S(X - \wt{X}) = (X - \wt{X}) / \|X - \wt{X}\|_2$, and $\otimes$ is the tensor product on the functional space $H$. Further, according to (\ref{whK}), our estimate of $K$ is
\begin{align}
\wh{K} = \frac{2}{N(N - 1)} \sum_{i < j} \{ S(X_i - X_j) \otimes S(X_i - X_j) \}. \nonumber
\end{align}
Let $\|\wh{K} - K\|_2 = \int \int_{\mathcal{T} \times \mathcal{T}} \{ \wh{K}(s, t) - K(s, t) \}^2 ds dt$ and $\|\wh{K} - K\|_{\infty} = \sup_{s, t \in \mathcal{T}} |\wh{K}(s, t) - K(s, t)|$.
Here $\wh{K}$ is an operator-valued U-statistics with kernel
\begin{align}
h(X_i, X_j) \triangleq S(X_i - X_j) \otimes S(X_i - X_j). \nonumber
\end{align}
It is easy to see that $Eh(X_i, X_j) = K$ and $\| h(X_i, X_j) \|_{\infty} \leq 1$. Then we make use of the asymptotic properties of operator-valued U-statistics discussed in \citep{sriperumbudur2017approximate} to derive our results.

Define $h_1(x_1) = E\{h(X_1, X_2) | X_1 = x_1\}$. Then we have the following decomposition
\begin{align}
\wh{K} - K = 2 \Big \{ \frac{1}{N} \sum_{i = 1}^N h_1(X_i) - K \Big \} + \Big [ \frac{1}{N(N - 1)} \sum_{i \neq j} \Big \{ h(X_i, X_j) - h_1(X_i) - h_1(X_j) + K \Big \} \Big ]. \nonumber
\end{align}
Thus,
\begin{align}\label{decomposition}
\| \wh{K} - K \|_2 \leq 2 \Big \| \frac{1}{N} \sum_{i = 1}^N h_1(X_i) - K \Big \|_2 + \Big \| \frac{1}{N(N - 1)} \sum_{i \neq j} \Big \{ h(X_i, X_j) - h_1(X_i) - h_1(X_j) + K \Big \} \Big \|_2.
\end{align}
For the first term, through Bernstein's inequality given in \citep{sriperumbudur2017approximate}, as $Eh_1(X_i) = K$ for $i = 1, \ldots, N$, there exists constants $C_1 > 0$ and $C_2 > 0$ such that for any $0 < \alpha < 1$,
\begin{align}
P \Bigg \{ \Big \| \frac{1}{N} \sum_{i = 1}^N h_1(X_i) - K \Big \|_2 \geq  \frac{C_1 \log \frac{2}{\alpha}}{N} + \sqrt{\frac{C_2 \log \frac{2}{\alpha}}{N}} \Bigg \} \leq \alpha. \nonumber
\end{align}
Thus we have
\begin{align}\label{part1}
\Big \| \frac{1}{N} \sum_{i = 1}^N h_1(X_i) - K \Big \|_2 = O_p(N^{-1/2}).
\end{align}
For the second term in (\ref{decomposition}), denote $h(X_i, X_j) - h_1(X_i) - h_1(X_j) + K$ as $W_{ij}$. Refer to \citep{sriperumbudur2017approximate}, there exists constant $C_3 > 0$ such that for any $\alpha > 0$,
\begin{align}\label{part2pro}
P \Big ( \Big \|\frac{1}{N(N - 1)} \sum_{i \neq j} W_{ij} \Big \|_2 \geq \frac{1}{N} \log \frac{C_3}{\alpha} \Big ) \leq \alpha.
\end{align}
That means
\begin{align}\label{part2}
\Big \| \frac{1}{N(N - 1)} \sum_{i \neq j} \Big \{ h(X_i, X_j) - h_1(X_i) - h_1(X_j) + K \Big \} \Big \|_2 = O_p(N^{-1}).
\end{align}
The deduction of (\ref{part2pro}) is analogous to that of (C.4) in \citep{sriperumbudur2017approximate}. Hence, we neglect the details here. Combining (\ref{decomposition}), (\ref{part1}) and (\ref{part2}), we have
\begin{align}
\| \wh{K} - K \|_2 = O_p(N^{-1/2}). \nonumber
\end{align}
As $\| \wh{K} - K \|_{\infty} \leq \| \wh{K} - K \|_2$, we then obtain (\ref{supKend}). The proof is completed.

\end{proof}

\begin{proof}[Proof of Theorem \ref{theoryeigfun}]

As the proof of Theorem 2 in \citep{yao2005functional} shown, the estimated eigenfunction shares the same convergence rate with the estimated Keandall's $\tau$ function. Therefore, according to the results in Theorem \ref{theoryKend}, Theorem \ref{theoryeigfun} is derived.

\end{proof}

\bibliographystyle{unsrtnat}
\bibliography{ref}

\end{document}